\documentclass[12pt]{amsart}
\usepackage{geometry} % see geometry.pdf on how to lay out the page. There's lots.
\usepackage{graphicx}
\usepackage{amsmath}
\usepackage{amssymb}

\newtheorem{mydef}{Definition}
\newtheorem{thm}{Theorem}
\newtheorem{lem}{Lemma}
\newtheorem{alg}{Algorithm}
\newtheorem{cor}{Corollary}

\title{
A Receding Horizon Strategy for Systems\\
with Interval-Wise Energy Constraints
}

\author{Eduardo Arvelo and Nuno C. Martins}
\date{}

\begin{document}

\begin{abstract}
\boldmath
We propose a receding horizon control strategy that readily handles systems that exhibit interval-wise total energy constraints on the input control sequence. The approach is based on a variable optimization horizon length and contractive final state constraint sets. The optimization horizon, which recedes by $N$ steps every $N$ steps, is the key to accommodate the interval-wise total energy constraints. The varying optimization horizon along with the contractive constraints are used to achieve analytic asymptotic stability of the system under the proposed scheme. The strategy is demonstrated by simulation examples.
\end{abstract}

\maketitle

\let\thefootnote\relax\footnotetext{This work is partially funded by ONR AppEl Center and the Multiscale Systems Center, one of six research centers funded under the Focus Center Research Program.}

\section{INTRODUCTION}

We are interested in designing a receding horizon control (RHC) algorithm that is able to readily accommodate interval-wise total energy constraints (ITEC). An ITEC system has the property that the total energy of the control input during a (pre-established) periodic time interval is limited (See Fig. \ref{description}). Examples of systems subject to ITEC include solar-powered systems that operate on battery power when sunlight is not available (e.g., at night). One can think of solar houses and solar-panel equipped mobile agents performing surveillance tasks for an extended period of time. Other examples of systems subject to ITEC can also be found in avionics and include temperature control of the aircraft cabin, for which power allocation may be limited at certain times, such as during take-off and landing.

The idea behind RHC is simple: the controller solves a finite-horizon control problem, and the first element of the computed control sequence is applied to the system. At the next time step (or sampling time for continuous-time systems), the procedure is repeated. The RHC strategy is very appealing because in most cases it allows for online tractable computation of the control input that does not suffer from the deficiencies of a pure open-loop strategy. The main drawback is that the direct implementation of RHC can be disastrous because stability is not (immediately) guaranteed.

Receding horizon control (or equivalently, model predictive control) dates back to the 1960s, but not until the late 1980s  was the issue of stability dealt rigorously \cite{Mayne:2000p442}. Until then, industry proponents of RHC had to manually ``fine tune" parameters to achieve a seemingly stable system. The first modification of RHC that guaranteed stability was the introduction of a zero final state constraint \cite{Mayne:1990p197},\cite{Keerthi:1988p725}, which in many cases renders the finite optimization problem unfeasible (or at the very least, it places a heavy burden on the controller). Later relaxations of the zero final state constraint included terminal cost functions (that penalize large final states) and terminal invariant constraint sets \cite{Chisci:1996p1035}, \cite{Scokeart:1999p1036}, \cite{Nicolao:2000p1307}. For a detailed survey of RHC, see \cite{Mayne:2000p442}.

Any of the traditional implementations of RHC (with step wise receding horizon) cannot readily accommodate the ITEC, as we explain in the next section. To remedy this problem, we implement a strategy that utilizes an initial horizon length of $2N$ that recedes by $N$ steps every $N$ steps (where $N$ is the length for which the ITEC is active) and does not recede at any other time instances. The basic idea is for the optimization horizon to encompass the ITEC in its entirety at all times. To guarantee stability of the system under this strategy, we impose contractive final state constraints during the optimization intervals for which the ITEC is not imposed, along with the assumption that the system is $\beta$-stabilizable (to be defined in section \ref{secbeta}).

The idea of contractive constraints was first suggested by \cite{Yang:1993p694} and has been later used in other works, such as \cite{Oliveira:1994p553} and \cite{Kothare:2000p452}. These papers also make use of a varying horizon strategy, which is also employed in \cite{Michalska:1993p201}, where a dual-mode controller is employed that steers the system to an invariant set by a horizon length that is computed online, and once the state lies in the invariant set the controller changes mode to drive it to zero. In \cite{Kothare:2000p452}, two horizons are employed: a shorter horizon used in the computation of the control sequence, and a longer horizon in which the predicted final state (propagated with constant input over the longer horizon) needs to satisfy the contractive constraint. The paper that is most similar to this one is \cite{Oliveira:1994p553}, where a smaller, non-receding contraction horizon is employed. However, the prediction horizon moves at every step and is not suitable for problems with systems subject to ITEC. Another paper worth to mention is \cite{Kothare:2003p686}, where a time-varying terminal constraint set is employed. Though the terminal constraint is not (necessarily) contractive and the authors employ a fixed horizon length, we feel this work is relevant to our research due to the time-varying nature of the terminal constraint set.

The paper is organized as follows: in section \ref{secprob}, we give the problem statement and discuss the shortcomings of the traditional RHC strategy applied to systems subject to ITEC. In section \ref{secbeta}, we introduce the idea of $\beta$-stabilizability. In section \ref{secvhc}, we describe the new RHC strategy, and provide proof of the stability of the system under the proposed algorithm. In section \ref{secsim}, we give a concrete example along with simulation results that showcases the new RHC strategy. In section \ref{seccon}, we conclude. Finally, in the appendix we provide the proof to a lemma used to prove our main result.

\begin{figure}[tb]
      \centering
      \includegraphics[width=12cm]{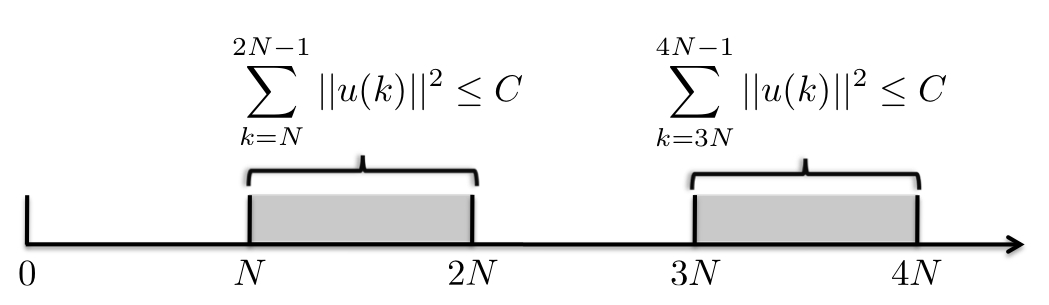}
      \caption{ITEC systems: Total energy constraints are enforced only at certain intervals}
      \label{description}
\end{figure}

%\section{Notation}

\section{Problem Statement}\label{secprob}
Suppose we have a system with dynamics given by:
\begin{equation}\label{dynamics}
x(k+1)= \phi(x(k), u(k)),
\end{equation}
\noindent where $x(k) \in \mathbb{R}^n$ and $u(k) \in \mathbb{R}^m$  are the state  and the input of the system at time $k$, respectively; and $\phi: \mathbb{R}^n \times \mathbb{R}^m\rightarrow \mathbb{R}^n$, with $\phi(0,0)=0$. Moreover, the system is constrained by:
\begin{align}
%x(k) & \in \mathcal{X}, ~~\forall k \label{c1}\\ 
u(k) &\in \mathcal{U},~~\forall k \label{c2} \\
\displaystyle \sum_{k=(2p-1)N}^{(2p)N-1}||u(k)||^2& \leq C, ~~p \in \mathbb{Z}_+, \label{c3}
\end{align}

\noindent where $\mathcal{U}$ is the input constraint set. Constraint (\ref{c3}) is known as the interval-wise total energy constraint (ITEC). In a real world application, this would be imposed during time intervals in which the system needs to rely solely on battery power. During other time intervals the system is powered by other energy sources (such as solar panels) and the total energy constraint is not imposed (see Fig. \ref{description}).

The infinite horizon control problem is to find the control sequence that minimizes the quadratic cost:
\begin{equation}
\begin{array}{ll}
\mathrm{minimize}~~&\displaystyle\sum_{k=1}^{\infty} ||x(k)||^2+||u(k-1)||^2\\
\mathrm{subject~to} & \mathrm{(\ref{dynamics}), (\ref{c2})~and~(\ref{c3}).}
\end{array}
\end{equation}
 
One way of generating control inputs in a tractable-online fashion is to employ a RHC strategy. The traditional implementation of RHC, however, causes a dilemma: how much energy should the controller allocate to the interval that partially covers the total energy constraint interval? One approach is to proportionally allocate the energy constraints. For example, with a horizon $N=3$, the traditional receding horizon approach would perform a finite optimization at $k=1$ (see Fig. \ref{rhc}) given by:
\begin{equation*}
\begin{array}{ll}
\mathrm{minimize}~~&\displaystyle\sum_{k=2}^{4} ||x(k)||^2+||u(k-1)||^2\\
\mathrm{subject~to} & \mathrm{(\ref{dynamics}), (\ref{c2})}\\
& ||u(3)||^2 \leq C/3, \label{c3_rhc1}
\end{array}
\end{equation*}
\noindent and at time $k=2$, the total energy constraint would become:
 \begin{equation*}
||u(3)||^2+||u(4)||^2 \leq 2C/3.
\end{equation*}
In this paper we propose a modified receding horizon strategy that not only provides analytical stability guarantees, but also does not require a heuristic solution to the ITEC allocation problem. It is important to note that the overlap of the constraint interval at the beginning of the optimization interval is no cause for concern. For example, at time $k=4$ in Fig. \ref{rhc} the amount of energy spent by $u(3)$ cannot be changed, say $||u(3)||^2=\gamma_3$. Then, the optimization will be carried out with the total energy constraint equal to:  
 \begin{equation*}
||u(4)||^2+||u(5)||^2 \leq C-\gamma_3.
\end{equation*}
\begin{figure}[tb]
      \centering
      \includegraphics[width=12cm]{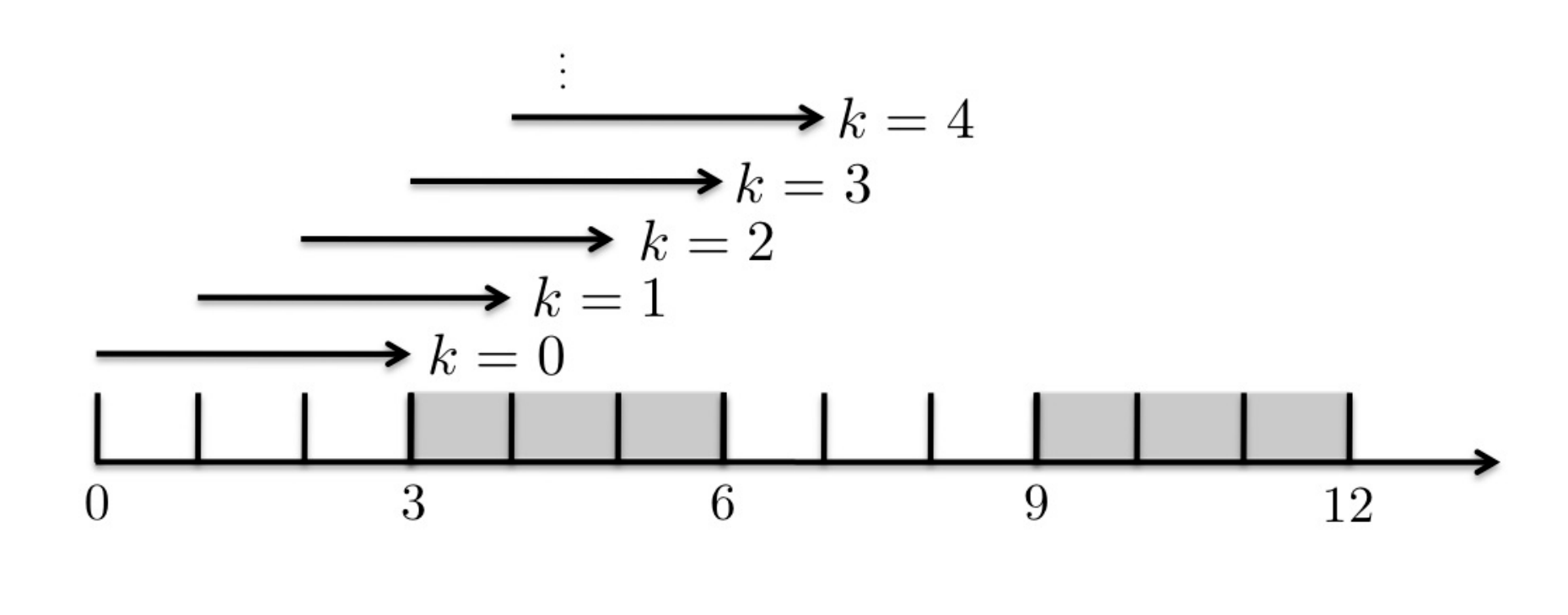}
      \caption{Traditional RHC with $N=3$. Note that during times $k=1$ and $k=2$, the total energy constraint interval is partially covered by the optimization horizon}
      \label{rhc}
\end{figure}

 \section{$\beta$-STABILIZABILITY}\label{secbeta}

The system is said to be $\beta$-stabilizable with $N$ steps if there exists an admissible control sequence that reduces the size of the state by a factor of $\beta$ in $N$ steps. In this paper, we do not require the contraction to occur over the optimization interval where the ITEC is active. A formal definition of $\beta$-stabilizability with ITEC is given below:

\begin{mydef} We say that a system $\phi$ is {$\beta$-stabilizable} with ITEC parameter $C$ and horizon $N$ in an open ball $B_N$ centered around the origin if there exists $\sigma>0$ such that:
\begin{align}
\mathcal{U}_N(x, \sigma, \beta) &\not= 0\\
\mathcal{U}^C_N(x, \sigma) &\not= 0
\end{align}
where $x \in B_N$, and $\mathcal{U}_N(x(k), \sigma, \beta)$, with $0 \leq \beta < 1$, is the set of all finite control sequences of length $N$ such that:
\begin{align}
\sum_{j=k+1}^{k+N} ||x(j)||^2+||u(j-1)||^2 & \leq \sigma ||x(k)||^2,\label{boundedcost}\\
u(k) &\in \mathcal{U} \label{boundedinput}\\
||x(k+N)||^2 &\leq \beta ||x(k)||^2
\end{align}
and $\mathcal{U}^C_N(x(k), \sigma)$ is the set of all finite control sequences of length $N$ such that (\ref{boundedcost}) and (\ref{boundedinput}) hold and 
\begin{align}
||x(k+N)||^2 &\leq ||x(k)||^2 \label{itecbound}\\
\displaystyle \sum_{j=k}^{k+N-1}||u(k)||^2& \leq C, ~~p \in \mathbb{Z}_+, \label{itec}
\end{align}
\end{mydef}\

Note that (\ref{itecbound}) does not impose contraction of the state norm if we constrain the control sequence to respect (\ref{itec}). We only require the state not to grow. Moreover, if the system has zero-input dynamics given by $\phi(x(k), 0) = x(k)$, then (\ref{itecbound}) is immediately satisfied.

 \section{INTERVAL-WISE HORIZON CONTROL} \label{secvhc}
 
In this section, we present an interval-wise receding horizon strategy (IRHC) that uses varying optimization horizons and contractive constraints. The varying horizon will not only enable the construction of an appropriate strategy for the problem with ITEC, but also guarantee asymptotic infinite horizon stability of the system with the imposition contractive constraints and the assumption that the system is $\beta$-stabilizable.

In order to accommodate for the total energy constraint, we employ a receding horizon strategy that entirely encompasses the constraint interval at all times (see Fig. \ref{vhc}). To accomplish this, the horizon is initially set for $h=2N$. Traditional RHC is performed, except that the horizon does ``recede" (i.e., $h = h-1$). When the horizon reaches $h=N+1$, it is expanded again to $h=2N$.
\begin{figure}[tb]
      \centering
      \includegraphics[width=12cm]{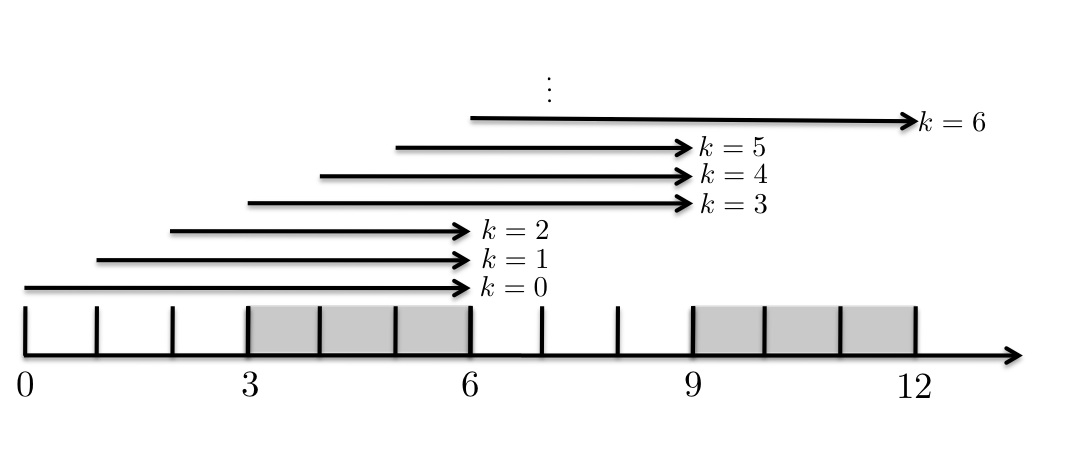}
      \caption{Interval-wise RHC with parameter $N=3$ based on algorithm \ref{algo1}}
      \label{vhc}
\end{figure}

We implement a receding horizon strategy in the following way:

\begin{alg}\label{algo1} {\bf IRHC with contractive constraints}
\begin{enumerate}
\item Set the optimization horizon $h = 2N$, time index $k=0$, energy-used tracking parameter $\gamma = 0$ and auxiliary indices $i=0$, $f=0$ and $p=N$.
\item Read in $x(k)$ and solve the constrained optimization problem:
\begin{equation}
\begin{array}{ll}
\mathrm{minimize}~~&\displaystyle\sum_{j=k+1}^{k+h} ||x(j)||^2+||u(j-1)||^2\\
\mathrm{subject~to} & \mathrm{(\ref{dynamics}), (\ref{c2})}\\
& \displaystyle \sum_{j=\max\{p,k\}}^{p+N-1}||u(j)||^2 \leq C-\gamma,\\
& ||x(k+h_c)||^2 \leq \beta^{(i+1)} ||x(0)||^2
\end{array}
\end{equation}
and apply the first element in the resulting control sequence to the system, say $u_1$.
\item If $f = 0$ (this means the ITEC sits at the beginning of the optimization horizon), set $\gamma=\gamma+||u_1||^2$
\item Reduce horizon $h = h-1$. 
\item If $h=N$, set $h=2N$ (this will push the horizon $N$ steps ahead) and $\gamma=0$ . If $k = p-1$, set $i=i+1$ (this will contract the constraint at the end of the new horizon), $f =1$; else if $k=p+N-1$, set $p=p+2N$ and $f =0$.
\item Set $k=k+1$ and go back to step 2.
\end{enumerate}
\end{alg}\

We are now ready to state our main result:

\begin{thm}\label{thm1} If the ITEC system $\phi$ is $\beta$-stabilizable with horizon $N$, then the interval-wise receding horizon strategy described in Algorithm \ref{algo1} will drive it asymptotically stable for any initial condition in $x(0) \in B_N$. Furthermore, the resulting dynamics will satisfy:
\begin{equation}
\sum_{j=1}^{\infty} ||x(j)||^2+||u(j-1)||^2 \leq \frac{1+\beta}{1-\beta} \sigma||x(0)||^2.
\end{equation}
\end{thm}\

We will prove this result by constructing a sequence of control sequences $v_q(k)$ indexed by $q$  that converges to the control sequence obtained by Algorithm \ref{algo1} and a bounded sequence $\Gamma_q$, associated with $v_q(k)$. To proceed, we need to define $v_q(k)$ and $\Gamma_q$:
\begin{mydef} $v_q(k)$ is a $q$-indexed sequence of control sequences defined by:
\begin{equation}\label{defv}
v_q(k) = \left\{ \begin{array}{ll} v_{q-1}(k) & 0\leq k< (q-1)N;\\
u_{qN}^{*(q+2)N-1} & (q-1)N \leq k< (q+1)N;\\
0 & k\geq (q+1)N\end{array} \right. 
\end{equation}
\noindent where $u_{k_1}^{*k_2}$ is the control segment obtained by Algorithm \ref{algo1} for $k_1 \leq k \leq k_2$, and $v_1(k)$ is given by:
\begin{equation}
v_1(k) = \left\{ \begin{array}{ll} u_{0}^{*2N-1} & 0\leq k< 2N;\\
0 & k \geq 2N \end{array} \right. 
\end{equation}
\end{mydef}\

In other words, $v_q(k)$ is indexed with each time the horizon is pushed $N$ steps ahead, and it consists of the past control inputs and the current $2N$ control segment, followed by zeros.

\begin{mydef} $\Gamma_q$ is defined as:
\begin{equation}\label{gamma}
\Gamma_q = \sum_{j=1}^{(q+1)N}||x(j)||^2+||v_q(j-1)||^2+\sum_{j=q+1}^{\infty}\beta^{\lceil \frac{j}{2} \rceil}\sigma||x(0)||^2
\end{equation}
\noindent where $v_q$ is as in (\ref{defv}) and $x(j)$ is generated by $v_q(j-1)$ and the dynamics (\ref{dynamics}), and $\lceil \cdot \rceil$ is the ceiling function.
\end{mydef}

Finally, the following lemma will be crucial to prove the asymptotic stability of the system under Algorithm \ref{algo1}:

\begin{lem}\label{lemma}If the system is $\beta$-stabilizable with horizon $N$, then:
\begin{equation}\label{lem1}
\Gamma_q \leq \frac{1+\beta}{1-\beta} \sigma||x(0)||^2, ~~\forall q\geq 1
\end{equation}
\noindent where $\Gamma_q$ is as in (\ref{gamma}).
\end{lem}
\begin{proof} [Proof of Lemma \ref{lemma}] See Appendix.
\end{proof}

\begin{proof}[Proof of Theorem \ref{thm1}] Let $G$ be defined as the infinite horizon cost incurred by the system under Algorithm \ref{algo1}, i.e.:

\begin{equation}
G \equiv \sum_{k=0}^{\infty} ||x(k+1)||^2+||u^*(k)||^2.\\
\end{equation}
\noindent and note that we can write:
\begin{align}
||u^*(k)||^2 &= ||u^*(k)-v_q(k)+v_q(k)||^2\\
&\leq  ||v_q(k)||^2+||u^*(k)-v_q(k)||^2;
\end{align}
\noindent and so, G is bounded by:
\begin{equation}
G \leq \lim_{q\to\infty}\sum_{k=0}^{\infty} ||x(k+1)||^2+ ||v_q(k)||^2+||u^*(k)-v_q(k)||^2
\end{equation}
Using (\ref{gamma}), the previous inequality can be written as:
\begin{align}
G &\leq \lim_{q\to\infty} \Gamma_q +\lim_{q\to\infty}\sum_{k=0}^{\infty}||u^*(k)-v_q(k)||^2\\
&= \lim_{q\to\infty} \Gamma_q \label{conv} \\
&\leq \frac{1+\beta}{1-\beta} \sigma||x(0)||^2\label{finalstep}
\end{align}
\noindent where (\ref{conv}) comes from the fact that the sequence $v_q(k)$ converges to sequence $u^*(k)$, and (\ref{finalstep}) comes from (\ref{lem1}). This concludes the proof.
\end{proof}\

Note that the IRHC algorithm can also be applied to system without ITEC, as described in the following corollary:
\begin{cor}If the non-ITEC system $\phi$ is $\beta$-stabilizable  with horizon $N$, then the IRHC strategy described in Algorithm \ref{algo1} (modified to add contraction at every $N$ steps) will drive it asymptotically stable for any initial condition in $x(0) \in B_N$. Furthermore, the resulting dynamics will satisfy:
\begin{equation}
\sum_{j=1}^{\infty} ||x(j)||^2+||u(j-1)||^2 \leq \frac{\sigma}{1-\beta} ||x(0)||^2
\end{equation}
\end{cor}\

The proof is similar to that of Theorem \ref{thm1}.

It is important to mention that the bound on the infinite horizon cost is only used to prove stability of the system and no claims on tightness are made. 

\section{SIMULATION}\label{secsim}
In this section, we employ the proposed algorithm to stabilize the two-dimensional nonlinear oscillator previously discussed in \cite{Primbs:1999p697}. The dynamics of the system are given by:

\begin{align*}
\dot{x}_1 = & x_2\\
\dot{x}_2 = &-x_1\Big(\frac{\pi}{2} +\tan^{-1}(5x_1)\Big)\\
& ~~~~~~~~-\frac{5x_1^2}{2(1+25x_1^2)}+4x_2+3u
\end{align*}

We used Matlab to implement the proposed IRHC algorithm to the discretized version of the system with sampling period $0.05$ and ITEC parameters $C=4.8$, $N=4$, initial state $x(0)=[2; -1]$ and $\beta=0.8$. The $\beta$-stabilizability (with $\beta=0.8$) was determined empirically. Figs. \ref{fig_x} and \ref{fig_u} show the evolution of the norms of the state and of the input, respectively, at different time instances. The solid lines (marked with x's) represent past signals while the dashed lines (marked with circles) represent the future prediction. Note that the ITEC constraints are satisfied and that the prediction horizon encompasses the constraint intervals in their entirety.  Fig. \ref{fig_x1x2} shows the plot of  $x_1$ vs. $ x_2$. Though the trajectory of the states have a oscillatory nature, the algorithm stabilizes the system.
\begin{figure}[tb]
      \centering
      \includegraphics[width=9.3cm]{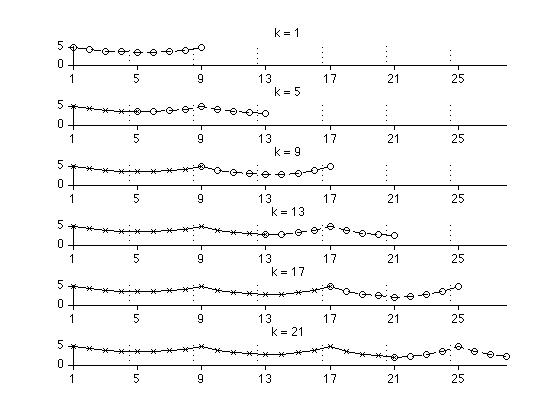}
      \caption{Evolution of the state norm for the ITEC simulation.}
      \label{fig_x}
          \includegraphics[width=9.3cm]{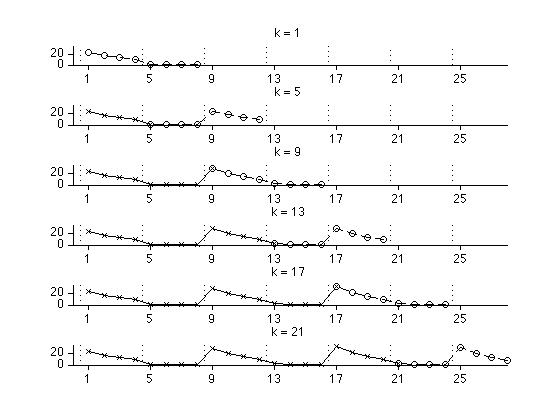}
           \caption{Evolution of the input norm for the ITEC simulation.}
                \label{fig_u}
\end{figure}

\begin{figure}[tb]
      \centering
      \includegraphics[width=9.3cm]{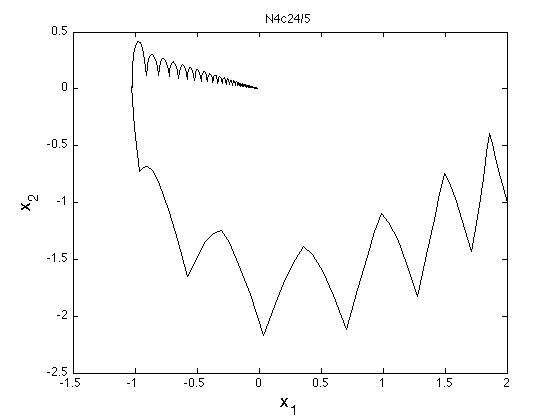}
      \caption{ITEC simulation: $x_1$ vs $x_2$}
      \label{fig_x1x2}
\end{figure}

Since this example was used in \cite{Primbs:1999p697} to showcase the dangers of the direct application of a receding horizon strategy, we made further simulations where the ITEC is not imposed. An interesting characteristic of this system (and the reason it was discussed in \cite{Primbs:1999p697}) is that the direct application of a RHC strategy works well for certain horizons, but is unstable for others. We ran simulations with $N=4$ and $N=5$ for the same initial condition as before. The traditional RHC used a fixed horizon of $2N$ and no endpoint constraints, while the proposed IRHC algorithm used contractive constraint parameter $\beta = 0.8$ and $\beta = 0.2$. We also compare the receding horizon strategies to the one obtained by $u=-3x_2$ (see \cite{Primbs:1999p697}). As we can see from Figs. \ref{ex4_1} and \ref{ex4_2}, the RHC performs very well for $N=4$, but the system becomes unstable when the horizon is increased to $N=5$. The proposed IRHC strategy with $N=4$ displays a similar oscillatory behavior as before, but incurs a lower cost. Table \ref{table1} gives the simulated infinite horizon costs for all cases. Note that the system with the constraint $\beta=0.8$ yields a lower cost than the system with the tighter constraint $\beta=0.2$, but convergence is faster. Moreover, it is interesting to note that when the horizon parameter is increased to $N=5$, the total cost for $\beta = 0.8$ increases while the cost for $\beta=0.2$ decreases. 

\begin{figure}[tb]
      \centering
      \includegraphics[width=9.3cm]{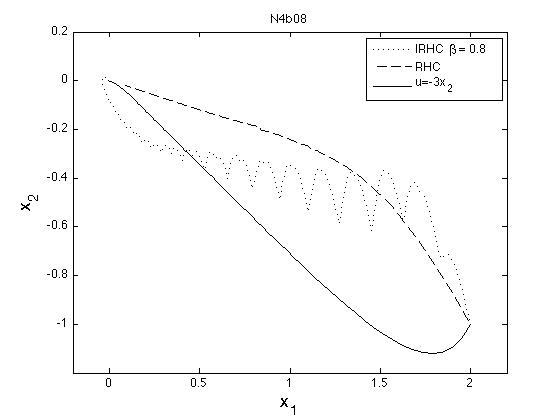}
      \caption{Non-ITEC simulation with parameter $N=4$, IRHC with $\beta =0.8$. Here the traditional RHC is stables and performs well}
      \label{ex4_1}
          \includegraphics[width=9.3cm]{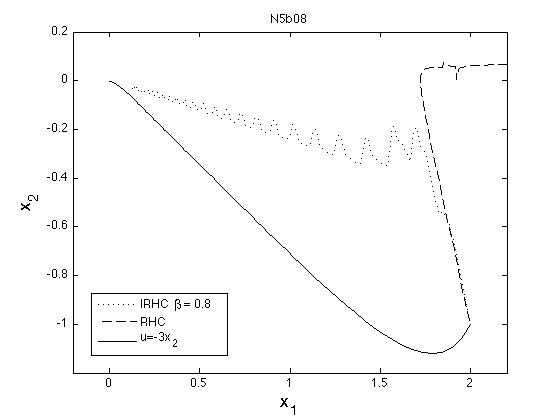}
           \caption{Non-ITEC simulation with parameter $N=5$, traditional RHC is unstable but IRHC $\beta =0.8$ with  still stabilizes the system}
           \label{ex4_2}
\end{figure}

\begin{table}[tb]
\caption{Performance Comparison }
\label{table1}
\begin{center}
\begin{tabular}{|r|c|c|}
\hline
& $N=4$ & $N=5$\\
\hline
$u=-3x_2$ &   		         \multicolumn{2}{|c|}{$325.4$} \\
\hline
RHC & 				$437.2$ & Unstable\\
\hline
IRHC($\beta = 0.8$) &      $412.9$ & 539.4\\
\hline
IRHC($\beta = 0.2$) & 	$3947$ &2053 \\
\hline
\end{tabular}
\end{center}
\end{table}

\section{CONCLUSION}\label{seccon}
We proposed a modified receding horizon strategy that readily handles total energy constraints that are only imposed at certain pre-specified periodic time intervals. Our approach utilizes an interval-wise receding horizon for the online optimization problem and contractive constraints to guarantee boundedness of the infinite horizon cost. The new algorithm is demonstrated by an example and compared to the direct implementation of the traditional step-wise receding horizon when no constraints are imposed. Future work include the use of imperfect state information (in a stochastic setting) and the application of this strategy to decentralized control problems. 

\appendix{}\label{secapp}
\begin{proof}[Proof of Lemma \ref{lemma}]
We will show that:
\begin{equation}
\Gamma_q \leq \frac{1+\beta}{1-\beta} \sigma||x(0)||^2, ~~\forall q\geq 1,
\end{equation}
\noindent by first showing that $\Gamma_{q+1}\leq\Gamma_q$, for all $q\geq 1$. Let 

\begin{equation}
\Gamma_q = \Gamma^1_q+\Gamma^2_q+\Gamma^3_q+\Gamma^4_q+\Gamma^5_q
\end{equation}
\noindent where:
\begin{align}
\Gamma^1_q &= \sum_{j=1}^{(q-1)N}||x(j)||^2+||v_q(j-1)||^2 \label{g1}\\
\Gamma^2_q &= \sum_{j=(q-1)N+1}^{qN}||x(j)||^2+||v_q(j-1)||^2\label{g2}\\
\Gamma^3_q &= \sum_{j=qN+1}^{(q+1)N}||x(j)||^2+||v_q(j-1)||^2\label{g3}\\
\Gamma^4_q & = \beta^{\lceil \frac{q+1}{2} \rceil}\sigma||x(0)||^2\label{g4}\\
\Gamma^5_q & = \sum_{j=q+2}^{\infty}\beta^{\lceil \frac{j}{2} \rceil}\sigma||x(0)||^2\label{g5}
\end{align}
From the definition of $v_q(k)$, we have that:
\begin{equation}\label{add1}
\Gamma_{q+1}^1=\Gamma_q^1+\Gamma^2_q.
\end{equation}

We also note that
\begin{equation}\label{add3}
\Gamma_{q+1}^2+\Gamma_{q+1}^3\leq \Gamma_{q}^3+ \Gamma_{q}^4
\end{equation}
\noindent from the definition of $v_q$ and the assumption that the system is $\beta$-stabilizable.

Moreover, from simple inspection of (\ref{g4}) and (\ref{g5}), we have that
\begin{align}
\Gamma^4_{q+1}+\Gamma^5_{q+1} &= \sum_{j=q+2}^{\infty}\beta^{\lceil \frac{j}{2} \rceil}\sigma||x(0)||^2\\
&= \Gamma^5_q.\label{add2}
\end{align}

Finally, using (\ref{add1}), (\ref{add3}), (\ref{add2}), we conclude that
\begin{align}
\Gamma_{q+1}&\leq\Gamma_q\\
		         &\leq \Gamma_1\\
		         &= \sum_{j=1}^{2N}||x(j)||^2+||u(j-1)||^2+\sum_{j=2}^{\infty}\beta^{\lceil \frac{j}{2} \rceil}\sigma||x(0)||^2\\		                  &\leq \sigma||x(0)||^2+\beta\sigma||x(0)||^2+\sum_{j=2}^{\infty}\beta^{\lceil \frac{j}{2} \rceil}\sigma||x(0)||^2\label{applybeta}\\
		         & =  \Big\{1+2\beta+2\sum_{j=2}^{\infty}\beta^{j} \Big\}\sigma ||x(0)||^2\\
		         &= \frac{1+\beta}{1-\beta} \sigma||x(0)||^2
\end{align}
\noindent where (\ref{applybeta}) follows from the definition of $\beta$-stabilizability.
\end{proof}

\bibliography{bibfile3}{}
\bibliographystyle{unsrt}

\end{document}